\documentclass[11pt,a4paper]{article}
\usepackage{vmargin}
\setmarginsrb{1.2in}{1.2in}{1.2in}{1.2in}{0mm}{0mm}{5mm}{5mm}

\usepackage{amsfonts,amsmath,amssymb}
\usepackage{amsthm}
\usepackage{graphicx}
\usepackage{color}

\newtheorem{theorem}{Theorem}
\newtheorem{lemma}{Lemma}
\newtheorem{corollary}{Corollary}

\newtheorem{proposition}{Proposition}
\newtheorem{observation}[theorem]{Observation}
\newtheorem{remark}{Remark}

\DeclareMathOperator{\maxNE}{maxNE}
\DeclareMathOperator{\minNE}{minNE}
\DeclareMathOperator{\OPT}{OPT}
\DeclareMathOperator{\min12cds}{Min-(1,2)-CDS}
\DeclareMathOperator{\minmkcds}{Min-(m,k)-CDS}

\title{On Selfish Creation of Robust Networks\\{\small (full version)}}
\author{Ankit Chauhan\thanks{Algorithm Engineering Group, Hasso Plattner Institute Potsdam, Germany, \texttt{firstname.lastname@hpi.de}} \and Pascal Lenzner\footnotemark[1] \and Anna Melnichenko\footnotemark[1] \and Martin M\"unn\thanks{Department of Computer Science, University of Liverpool, UK, \texttt{M.F.Munn@liverpool.ac.uk}}}

\date{}
\begin{document}
\maketitle
\begin{abstract}
  \noindent Robustness is one of the key properties of nowadays networks. However, robustness cannot be simply enforced by design or regulation since many important networks, most prominently the Internet, are not created and controlled by a central authority. Instead, Internet-like networks emerge from strategic decisions of many selfish agents. Interestingly, although lacking a coordinating authority, such naturally grown networks are surprisingly robust while at the same time having desirable properties like a small diameter. 
  To investigate this phenomenon we present the first simple model for selfish network creation which explicitly incorporates agents striving for a central position in the network while at the same time protecting themselves against random edge-failure. We show that networks in our model are diverse and we prove the versatility of our model by adapting various properties and techniques from the non-robust versions which we then use for establishing bounds on the Price of Anarchy. Moreover, we analyze the computational hardness of finding best possible strategies and investigate the game dynamics of our model.             
\end{abstract}

\section{Introduction}
Networks are everywhere and we crucially rely on their functionality. Hence it is no surprise that designing networks under various objective functions is a well established research area in the intersection of Operations Research, Computer Science and Economics. However, investigating how to create suitable networks \emph{from scratch} is of limited use for understanding most of nowadays networks. The reason for this is that most of our resource, communication and online social networks have not been created and designed by some central authority. Instead, these critical networks emerged from the interaction of many selfish agents who control and shape parts of the network. This clearly calls for a game-theoretic perspective.

One of the most prominent examples of such a selfishly created network is the Internet, which essentially is a network of sub-networks which are each owned and controlled by Internet service providers (ISP). Each ISP decides selfishly how to connect to other ISPs and thereby balancing the cost for creating links (buying the necessary hardware and/or peering agreement contracts for routing traffic) and the obtained service quality for its customers. Interestingly, although the Internet is undoubtedly an important and critical infrastructure, there is no central authority which ensures its functionality if parts of the network fail. Despite this fact, the Internet seems to be robust against node or edge failures, which hints that a socially beneficial property like network robustness may emerge from selfish behavior.  

Modeling agents with a desire for creating a \emph{robust} network has long been neglected and was started to be investigated only very recently. 
This paper contributes to this endeavor by proposing and analyzing a model of selfish network creation, which explicitly incorporates agents who strive for occupying a central position in the network while at the same time ensuring that the overall network remains functional even under edge-failure.

\subsection{Related Work}
Previous research on game-theoretic models for network creation has either focused on  \emph{centrality models}, where the agents' service quality in the created network depends on the distances to other agents, or on \emph{reachability models} where agents only care about being connected to many other agents. 

Some prominent examples of centrality models for selfish network creation are \cite{JW96,Fab03,CP05,De07,ADHL13,MMO14,Ehs15}. They all have in common that agents correspond to nodes in a network and that the edge set of the network is determined by the combination of the agents' strategies. The utility function of an agent contains a service quality term which depends on the distances to all other agents. 
To the best of our knowledge the very recent paper by Meirom et al.~\cite{MMO15} is the only centrality model which incorporates edge-failures. The authors consider two types of agents, major-league and minor-league agents, which maintain that the network remains $2$-connected while trying to minimize distances, which are a linear combination of the length of a shortest path and the length of a best possible vertex disjoint backup path. Under some specific assumptions, e.g. that there are significantly more minor-league than major-league agents, they prove various structural properties of equilibrium networks and investigate the corresponding game-dynamics. In contrast to this, we will investigate a much simpler model with homogeneous agents which is more suitable for analyzing networks created by equal peers. Our results can be understood as zooming in on the sub-network formed by the major-league agents (i.e. top tier ISPs).        

In reachability models, e.g. \cite{BG00,Kli11,BGP15,FGMMM12,ABU15,GJKKM15}, the service quality of an agent is simply defined as the number of reachable agents and distances are ignored completely. For reachability models the works of Kliemann~\cite{Kli11,Kli13} and the very recent paper by Goyal et al.~\cite{GJKKM15} explicitly incorporate a notion of network robustness in the utility function of every agent. All models consider an external adversary who strikes after the network is built. In~\cite{Kli11,Kli13} the adversary randomly removes a single edge and the agents try to maximize the expected number of reachable nodes post attack. Two versions of the adversary are analyzed: edge removal uniformly at random or removal of the edge which hurts the society of agents most. For the former adversary a constant Price of Anarchy is shown, whereas for the latter adversary this positive result is only true if edges can be created unilaterally. In~\cite{GJKKM15} nodes are attacked (and killed) and this attack spreads virus-like to neighboring nodes unless they are protected by a firewall. Interestingly also this model has a low Price of Anarchy and the authors prove a tight linear bound on the amount of edge-overbuilding due to the adversary.      

\subsection{Model and Notation}
We consider the Network Creation Game (NCG) by Fabrikant et al.~\cite{Fab03} augmented with the uniform edge-deletion adversary from Kliemann~\cite{Kli11} and we call our model \emph{Adversary NCG} (Adv-NCG). More specifically, in an Adv-NCG there are $n$ selfish agents which correspond to the nodes of an undirected multi-graph $G = (V,E)$ and we will use the terms agent and node interchangeably.
A pure strategy $S_u$ of agent $u\in V(G)$ is any multi-set over elements from $V\setminus\{u\}$. If $v$ is contained $k$ times in $S_u$ then this encodes that agent $u$ wants to create $k$ undirected edges to node $v$. Moreover we say that $u$ is the owner of all edges to the nodes in $S_u$. We emphasize the edge-ownership in our illustrations by drawing directed edges which point away from their owner - all edges are nonetheless understood to be undirected. Given an $n$-dimensional vector of pure strategies for all agents, then the union of all the edges encoded in all agents' pure strategies defines the edge set $E$ of the multi-graph $G$. Since there is a bijection of multi-graphs with edge-ownership information and pure strategy-vectors, we will use networks and strategy-vectors interchangeably, e.g. by saying that a network is in equilibrium.       

The agents prepare for an adversarial attack on the network after creation. This attack deletes one edge uniformly at random. Hence, agents try to minimize the attack's impact on themselves by minimizing their \emph{expected cost}. Let $G-e$ denote the network $G$ where edge $e$ is removed. Let $\delta_G(u) = \sum_{v\in V(G)} d_G(u,v)$, where $d_G(u,v)$ is the number of edges of a shortest path from $u$ to $v$ in network $G$. Let $$dist_G(u) = \frac{1}{|E|}\sum_{e \in E} \delta_{G-e}(u) = \frac{1}{|E|} \sum_{e\in E}\sum_{v \in V} d_{G-e}(u,v)$$ denote agent $u$'s \emph{expected distance cost} after the adversary has removed some edge uniformly at random from $G$. The expected cost of an agent $u$ in network $G = (V,E)$ with edge-price $\alpha$ is defined as $cost_u(G,\alpha) = edge_u(G,\alpha) + dist_G(u)$, where $edge_u(G,\alpha) = \alpha \cdot |S_u|$ is the total edge-cost for agent $u$ with strategy $S_u$ in $(G,\alpha)$. Thus, compared to the NCG~\cite{Fab03}, the distance cost term is replaced by the expected distance cost with respect to uniform edge deletion.

Let $(G,\alpha)$ be any network with edge-ownership information. We call any strategy-change from $S_u$ to $S_u'$ of some agent $u$ a \emph{move}.  Specifically, if $|S_u| = |S_u'|$, then such a move is called a \emph{multi-swap}, if $|S_u \cap S_u'| < |S_u|-2$ and a \emph{swap} if $|S_u \cap S_u'| = |S_u|-2$. If a move of agent $u$ strictly decreases agent $u$'s cost, then it is called an \emph{improving move}. If no improving move exists, then we say that agent $u$ plays its \emph{best response}. Analogously we call a strategy-change towards a best response a \emph{best response move}. A sequence of best response moves which starts and ends with network $(G,\alpha)$ is called a \emph{best response cycle}. We say that $(G,\alpha)$ is in \emph{Pure Nash Equilibrium} (NE) if all agents play a best response.  

We measure the overall quality of a network $(G,\alpha)$ with its \emph{social cost}, which is defined as $cost(G,\alpha) = \sum_{u \in V(G)} cost_u(G,\alpha) = edge(G,\alpha) + dist(G)$, where $edge(G,\alpha) = \sum_{u \in V(G)} edge_u(G,\alpha) = \alpha\cdot|E|$ and $dist(G) = \sum_{u \in V(G)} dist_G(u)$. Let $OPT(n,\alpha)$ be a network on $n$ nodes with edge-price $\alpha$ which minimizes the social cost and we call $\OPT(n,\alpha)$ the \emph{optimum network} for $n$ and $\alpha$. Let $\maxNE(n,\alpha)$ be the maximum social cost of any NE network on $n$ agents with edge-price $\alpha$ and analogously let $\minNE(n,\alpha)$ be the NE having minimum social cost. Then, the \emph{Price of Anarchy} is the maximum over all $n$ and $\alpha$ of the ratio $\tfrac{\maxNE(n,\alpha)}{\OPT(n,\alpha)}$, whereas the \emph{Price of Stability} is the maximum over all $n$ and $\alpha$ of the ratio $\tfrac{\minNE(n,\alpha)}{OPT(n,\alpha)}$.

\subsection{Our Contribution}
This paper introduces and analyzes an accessible model, the Adv-NCG, for selfish network creation in which agents strive for a central position in the network while protecting against random edge-failures. 

We show that optimum networks in the Adv-NCG are much more diverse than without adversary, which also indicates that the same holds true for the Nash equilibria of the game. However, we also show that many techniques and results from the NCG can be adapted to cope with the Adv-NCG, which indicates that the influence of the adversary is limited. In particular, we prove NP-hardness of computing a best possible strategy and W[2]-hardness of computing a best multi-swap. Moreover, we show that the Adv-NCG is not weakly acyclic, which is the strongest possible non-convergence result for any game. On the positive side, we prove that the amount of edge-overbuilding due to the adversary is limited, which is then used for proving that upper bounding the diameter essentially bounds the Price of Anarchy from above. We apply this by adapting two diameter-bounding techniques from the NCG to the adversarial version which then yields an upper bound on the PoA of $\mathcal{O}(1+\alpha/\sqrt{n})$.

\section{Optimal Networks}\label{sec_OPT}
Clearly, every optimal network must be $2$-edge connected. Thus, every optimal network must have at least $n$ edges. 
We first prove the intuitive fact that if edges get more expensive, then the optimum networks will have fewer edges.
\begin{theorem}\label{thm_optedges}
 Let $(G = (V,E),\alpha)$ and $(G' = (V,E'),\alpha')$ be optimal networks on $n$ nodes in the Adv-NCG for $\alpha$ and $\alpha'$, respectively. If $\alpha' > \alpha$, then $|E|\geq |E'|$.  
\end{theorem}
\begin{proof}
 We prove the statement by contradiction. Assume that $\alpha' > \alpha$ holds and that network $G' = (V,E')$ has strictly more edges than network $G=(V,E)$. Let $\Delta = |E'|-|E|$ denote this difference. 
 
 On the one hand, since $(G',\alpha')$ is an optimal network for edge price $\alpha'$, we have the social cost of $(G',\alpha')$ must be at most the social cost of $(G,\alpha')$. Thus we have 
 \begin{align*}
  cost(G',\alpha') &\leq cost(G,\alpha') \\
  \iff edge(G',\alpha') + dist(G') &\leq edge(G,\alpha') + dist(G)\\
  \iff edge(G',\alpha') - edge(G,\alpha') &\leq dist(G) - dist(G')\\
  \iff \Delta\alpha' &\leq dist(G) - dist(G').
 \end{align*}
On the other hand, since $(G,\alpha)$ is an optimal network for edge price $\alpha$, we have 
\begin{align*}
  cost(G,\alpha) &\leq cost(G',\alpha) \\
  \iff edge(G,\alpha) + dist(G) &\leq edge(G',\alpha) + dist(G')\\
  \iff dist(G) - dist(G') &\leq  edge(G',\alpha) - edge(G,\alpha)\\
  \iff dist(G) - dist(G') &\leq \Delta\alpha.
 \end{align*}
 Hence, we have 
  $\Delta\alpha' \leq dist(G) - dist(G') \leq \Delta\alpha$, which implies $\alpha' \leq \alpha$ which contradicts our assumption that $\alpha' > \alpha$. 
\end{proof}
\begin{remark}
 Note that the above proof works for any function $dist_G:V_G\rightarrow\mathbb{R}^+$, that is, in particular also for the NCG~\cite{Fab03}.
\end{remark}
\noindent In the following, we show that the landscape of optimum networks is much richer in the Adv-NCG, compared to the NCG where the optimum is either a clique or a star, depending on $\alpha$. In particular, we prove that there are $\Omega(n^2)$ different optimal topologies. We consider the following types of networks:
\begin{figure}[h!]
 \centering
 \includegraphics[width=\textwidth]{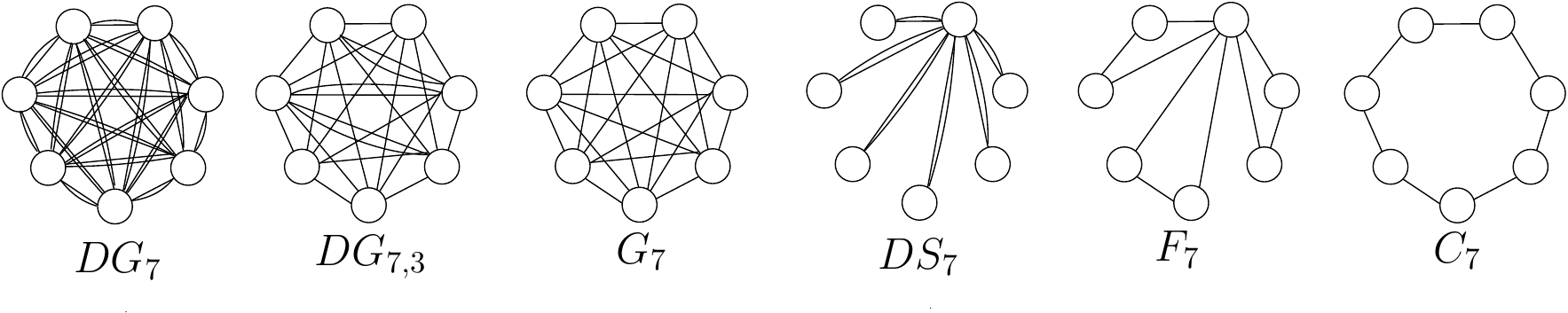}
 \caption{Different candidates for optimum networks.}
 \label{fig:graphtypes}
\end{figure}
Here $DG_n$ is a clique of $n$ nodes where we have a double edge between all pairs of nodes. Let $DG_{n,k}$ be a $n$ node clique with exactly $k$ pairs of nodes which are connected with double edges. Thus, $DG_{n,0} = G_n$ and $DG_{n,{n\choose 2}} = DG_n$. Moreover, let $F_n$ denote the fan-graph on $n$ nodes which is a collection of triangles which all share a single node and let $DS_n$ denote a star on $n$ nodes where all connections between the center and the leaves are double edges. Finally, let $C_n$ be a cycle of length~$n$. 

Clearly, if $\alpha = 0$, then the optimum network on $n$ nodes must be a $DG_n$, since in this network no edge deletion of the adversary has any effect, it minimizes the sum of distances of each agent and since edges are for free.

Now consider what happens, if one pair of agents, say $u$ and $v$, are just connected via a single edge instead of a double edge. The probability that the adversary removes this edge is $\frac{1}{n(n-1)-1}$. The removal would cause an increase in distance cost of $1$ between $u$ and $v$ and between $v$ and $u$. Thus, if $\alpha < \frac{2}{n(n-1)-1}$, then agent $u$ and $v$ would individually be better off buying another edge between each other. Thus, we have the following observation.
\begin{observation}\label{obs_opt}
 If $0 \leq \alpha \leq \frac{2}{n(n-1)-1}$, then $\OPT(n,\alpha) = DG_n$. 
\end{observation}

\begin{lemma}\label{lem_opt_range}
 If $\frac{2n(n-1)}{({n\choose 2}+k)({n \choose 2}+k+1)} \leq \alpha \leq \frac{2n(n-1)}{({n\choose 2}+k)({n \choose 2}+k-1)}$, for $1 \leq k \leq {n \choose 2}-1$, then the network $DG_{n,k}$ is optimal.
\end{lemma}
\begin{proof}
Consider network $DG_{n,k}$ with k double edges. Let $u_i$ denote a node with exactly $i$ incident single edges. We have $$cost_G(u_i)=edge_G(u_i)+\frac{(|E|-i)(n-1)+i\cdot n}{|E|}=edge_G(u_i)+(n-1)+\frac{i}{|E|},$$ because the distance between $u_i$ and any other node increases only if the adversary deletes any of $u_i$'s incident $i$ single edges.
 The social cost of $DG_{n,k}$ is $cost(DG_{n,k})=\sum_{i=0}^{{n\choose 2}-k}{a_i\cdot cost_G(u_i)}$, where $a_i$ is the number of vertices having exactly $i$ incident single edges. Note, that $\sum_{i=0}^{{n\choose 2}-k}{a_i}=n$. 
Thus, $cost(DG_{n,k})$ is $$
\alpha|E|+\sum_{i=0}^{{n\choose 2}-k}{a_i}\left(n-1+\frac{i}{|E|}\right)
			  =\alpha|E|+n(n-1)+\frac{1}{|E|}\sum_{i=0}^{{n\choose 2}-k}{a_i\cdot i}.
$$
Now we simplify the above cost-function. Consider the induced sub-graph $G=(V,E')$ of network $DG_{n,k}$ which contains only the single edges of $DG_{n,k}$. By using the Handshake Lemma we obtain $\sum_{i=0}^{{n\choose 2}-k}{a_i\cdot i}=\sum_{u\in V}{deg_G (u)}=2|E'|=2({n\choose 2}-k)$, where $deg_G(u)$ is $u$'s degree in $G$.
Hence, we have $$cost(DG_{n,k})=\alpha|E|+n(n-1)+\frac{2({n\choose 2}-k)}{{n\choose 2}+k}.$$
Now, if look at the cost difference between $DG_{n,k-1}$ and $DG_{n,k}$ we get
$$cost(DG_{n,k}) - cost(DG_{n,k-1}) = \alpha -  \frac{2({n\choose 2}-k+1)}{{n\choose 2}+k-1}+\frac{2({n\choose 2}-k)}{{n\choose 2}+k}.$$ The cost difference between $DG_{n,k}$ and $DG_{n,k+1}$ is
$$cost(DG_{n,k}) - cost(DG_{n,k+1}) = -\alpha - \frac{2({n\choose 2}-k)}{{n\choose 2}+k}+\frac{2({n\choose 2}-k-1)}{{n\choose 2}+k+1}.$$
Thus, if $\alpha  = \alpha_k$ where $$\frac{2({n\choose 2}-k)}{{n\choose 2}+k}-\frac{2({n\choose 2}-k-1)}{{n\choose 2}+k+1} \leq\alpha_k \leq \frac{2({n\choose 2}-k+1)}{{n\choose 2}+k-1}-\frac{2({n\choose 2}-k)}{{n\choose 2}+k}, $$ then upgrading a single edge to a double or downgrading a double edge to a single in $DG_{n,k}$ does not decrease the social cost. Since $\alpha_k < \frac{8}{n^2} < 2 - \frac{2}{n} < 2 - \frac{2}{|E|+1}$, it follows that $\OPT(n,\alpha_k)$ has diameter $1$, since otherwise there are two agents $u$ and $v$ having expected distance at least $2$ between each other and inserting the edge $uv$ would decrease their expected distance to $1 + \frac{1}{|E|+1}$ which yields a decrease in social distance cost of at least $2 \left(2 - (1 + \frac{1}{|E|+1})\right) = 2 - \frac{2}{|E|+1}$. 

Since $OPT(n,\alpha_k)$ has diameter $1$ we know that we can obtain $OPT(n,\alpha_k)$ from $DG_{n,k}$ by either downgrading some double edges to single edges or by upgrading some single edges to double edges. We have chosen $\alpha_k$ such that downgrading one double edge to a single edge or upgrading one single edge to a double edge does not decrease the social cost. Since downgrading or upgrading edges only affects the distance costs of the incident agents it follows that if downgrading or upgrading one edge does not decrease the social cost, then downgrading or upgrading more than one edge cannot decrease the social cost as well.  
Thus, $DG_{n,k}$ is the optimal network for $\alpha_k$. 
\end{proof}
\noindent Note, that the proof of the above statement implies that the complete graph $G_n$ is an optimum, if $\frac{4}{{n\choose 2}+1} \leq \alpha < 2-\frac{2}{{n \choose 2}} $.

We also remark that we conjecture that Fig.~\ref{fig:graphtypes} resembles a snapshot of optimum networks for increasing $\alpha$ from left to right. In fact, extensive simulations indicate that the optimum changes from $G_n$ to $DS_n$ and then, for slightly larger $\alpha$ for $F_n$. After this the cycles in the fan-graph increase and get fewer in number until, finally, for $\alpha \in \Omega(n^3)$ the cycle appears as optimum. 

\section{Computing Best Responses and Game Dynamics}\label{sec_BR}
In this section we investigate computational aspects of the Adv-NCG. First we analyze the hardness of computing a best response and the hardness of computing a best possible multi-swap. Then we analyze a natural process for finding an equilibrium network by sequentially performing improving moves.
\subsection{Hardness of Best Response Computation}
We first introduce useful properties for ruling out multi-buy or multi-delete moves. The proof is similar to the proof of Lemma 1 in~\cite{Len12}. 
\begin{proposition}\label{pro1}
If an agent cannot decrease its expected cost by buying (deleting) one edge in the Adv-NCG, then buying (deleting) $k > 1$ edges cannot decrease the agent's expected cost.
\end{proposition}
\begin{lemma}
\label{lemBR}
If $1-\frac{1}{|E|+1}<\alpha<1 + \frac{1}{|E|(|E|-1)}$ and if agent $u$ is not an endpoint of any double-edge in the Adv-NCG, then buying the minimum number of edges such that $u$'s expected distance to all nodes in $V\setminus N_u$ is 2 and to nodes in $N_u$ is $1+\frac{1}{|E|}$ is $u$'s best response.
\end{lemma}
\begin{proof}
Consider any network $(G,\alpha)$ where $u$'s expected distance to all nodes in $V\setminus N_u$ is $2$ and to all nodes in $N_u$ is $1 + \frac{1}{|E|}$. 

Buying an additional edge to some $v\in N_u$ in $G$ decreases $u$'s expected distance to $v$ by $\frac{1}{|E|}$. Buying an edge towards a node $w \notin N_u$ decreases $u$'s expected distance to $w$ by $1-\frac{1}{|E|+1} > \frac{1}{|E|}$. Thus if $\alpha >1-\frac{1}{|E|+1}$, then buying a single edge does not decrease $u$'s expected cost. Thus by Proposition~\ref{pro1}, agent $u$ cannot improve its expected cost in $G$ by buying more than one edge.

Swapping an edge to some $v\in N_u$ decreases $u$'s expected distance to $v$ by $\frac{1}{|E|}$ but increases $u$'s expected distance to some $w\in N_u$ by $1-\frac{1}{|E|}$. Swapping an edge towards a node $w \notin N_u$ decreases $u$'s expected distance to $w$ by $1-\frac{1}{|E|}$ but increases the expected distance to $w\in N_u$ by at least $1-\frac{1}{|E|}$.

If $u$ has bought the minimum number of edges such that $u$'s expected distance to all nodes in $V\setminus N_u$ is $2$ and to all nodes in $N_u$ is $1 + \frac{1}{|E|}$, then
deleting an edge from some $v\in N_u$ increases $u$'s expected distance to $v$ by at least $1 + \frac{1}{|E|(|E|-1)}$ since after deleting the edge the expected distance between $u$ and $v$ is $2+\frac{1}{|E|-1}$. Thus, if $\alpha< 1 + \frac{1}{|E|(|E|-1)}$, then deleting a single edge does not decrease $u$'s expected cost. Thus, by Proposition~\ref{pro1}, agent $u$ cannot decrease its expected cost by deleting more than one edge.
\end{proof}
\noindent Now we show that computing the best possible strategy-change is intractable.
\begin{theorem}\label{thm_BRhardness}
~
\begin{itemize}
\item[1.] It is NP-hard to compute the best response of agent $u$ in the Adv-NCG.  
\item[2.] It is $W[2]$-hard to compute the best multi-swap of agent $u$ in the Adv-NCG. 
\end{itemize}
\end{theorem}
\begin{proof}
We prove both statements by reduction from \textsc{Minimum-$m$-connected $k$-dom\-ina\-ting set}($\minmkcds$)~\cite{Shang07} which is defined as follows:  Given a graph $G=(V,E)$ and two natural numbers $m$ and $k$, find a subset $S\subseteq V$ of minimum size such that every vertex in $V/S$ is adjacent to at least $k$ vertices in $S$ and the induced sub-graph of $S$ is $m$-connected. 

\noindent (1) For the hardness we give the reduction by $\min12cds$ which is a NP-hard problem as a consequence of Theorem~\ref{Thm_Hardness}. 
Given the configuration of rest of the graph, agent $u$ has to pick subset of vertices such that $cost_G(u)$ is minimized. For any $1-\frac{1}{|E|+1}<\alpha<1 + \frac{1}{|E|(|E|-1)}$  and if there are no other agent has bought an edge to $u$, then, by Lemma~\ref{lemBR}, the best response is to buy edges to all the nodes in the $S$. In that case the expected eccentricity of $u$ will be at most 2. Since every node $w\notin S$ is adjacent to at least two vertices in $S$, the distance to $u$ cannot increase due to single edge-deletion. Since $S$ is connected, $u$ has expected distance at most $1+\frac{1}{|E|}$ to all nodes in $S$.

\noindent (2) Consider that agent $u$ has budget $b$, then it follows from Lemma~\ref{lemBR} and the above proof of (1), that the best response will be to create edges to the vertices in $S$ of size at most $b$. By Theorem~\ref{Thm_Hardness} (See Appendix) $\min12cds$ parameterized by the size of the set $S$ is W[2]-hard. 

Now, if there exists a $\min12cds$ $S$ of size at most $b$ then $S$ must be a subset of $u$'s best response. Hence if we know the best response of agent $u$ then we can solve the parameterized version of $\min12cds$ of $G$ by checking all possible subsets of $u$'s best response in $2^b\mathcal{O}(n)$ time.  
\end{proof}

\subsection{Game Dynamics}
We investigate the dynamic properties of the Adv-NCG. That is, we turn the model into a sequential version which starts with some initial network $(G,\alpha)$ and then agents move sequentially in some order and perform improving moves, if possible. One natural question is, if this process is guaranteed to converge to a Nash equilibrium of the game.     

For the game dynamics of the Adv-NCG we prove the strongest possible negative result, which essentially shows that there is no hope for convergence if agents stick to performing improving moves only. In particular, we prove that the order of the agents moves or any tie-breaking between different improving moves does not help for achieving convergence. This result is even stronger than the best known non-convergence results for the NCG~\cite{KL13}.  
\begin{theorem}\label{thm_notWA}
 The $Adv-NCG$ is not weakly acyclic. 
\end{theorem}
\begin{proof}[Proofsketch]
We prove the statement by giving a best response cycle $(G_1,\alpha),\dots,(G_7,\alpha)$, where $(G_1,\alpha) = (G_7,\alpha)$ and $(G_{i+1},\alpha)$ is obtained from $(G_i,\alpha)$ by an improving move of one agent in $(G_i,\alpha)$. Our best response cycle has the special features that in every step of the cycle there is exactly one agent who can perform an improving move and that this improving move is unique. Thus, starting with $(G_1,\alpha)$, any sequence of improving moves must be infinite. The best response cycle on $10$ agents with $\alpha = 10.3$ is depicted in Fig.~\ref{fig:model1_notweaklyacyclic}. 
\begin{figure}[h!]
 \centering
 \includegraphics[width=0.9\textwidth]{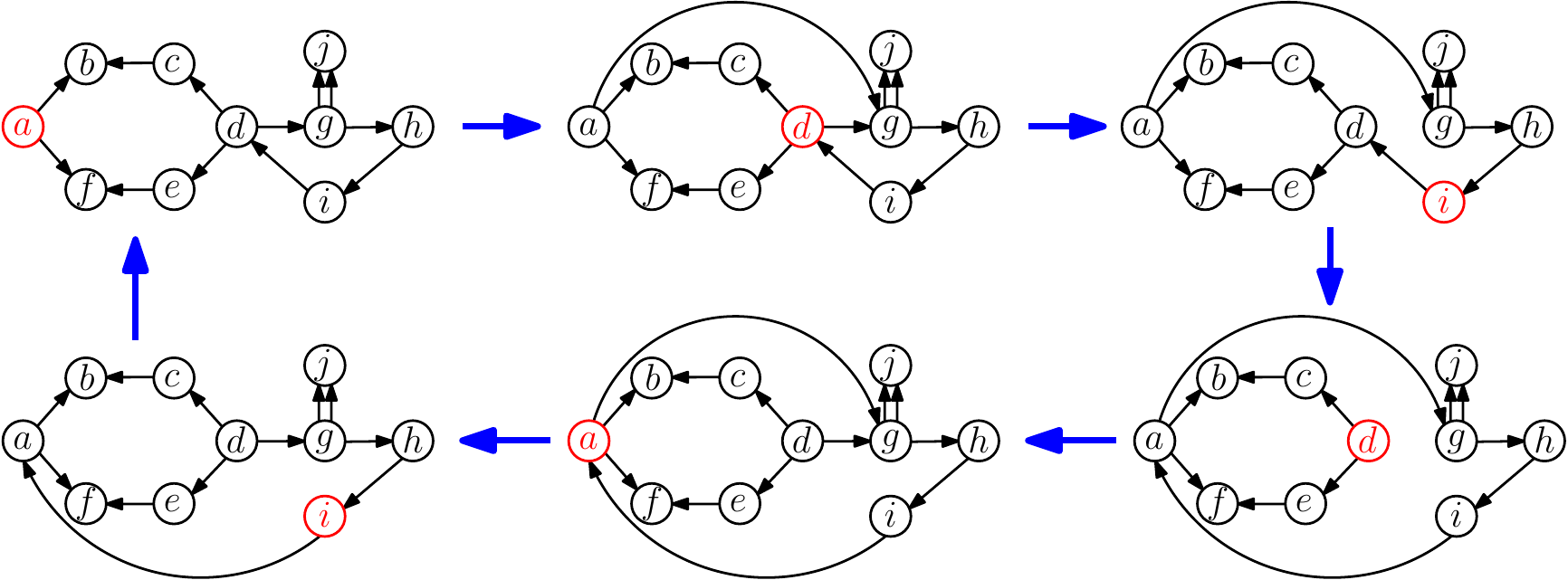}
 \caption{A best response cycle for $\alpha = 10.3$, where in every step only the red agent has an improving move and this improving move is unique.}
 \label{fig:model1_notweaklyacyclic}
\end{figure}
We omit the quite lengthy proof that the shown best response cycle behaves as indicated. The interested reader can get more details and a Python-script performing an exhaustive search on all possible strategy-changes from the authors.  
\end{proof}

\section{Analysis of Networks in Nash Equilibrium}\label{sec_NE}
In this section we establish the existence of networks in Nash Equilibrium for almost the whole parameter space and we compare NE networks in the Adv-NCG with NE networks from the NCG~\cite{Fab03} and Kliemann's adversarial model~\cite{Kli11}.  
Moreover, we investigate structural properties which allow us to provide bounds on the Price of Stability and the Price of Anarchy.

We start with the existence result:
\begin{theorem}\label{thm_existence}
 The networks $DG_n$ and $DS_n$ are in pure Nash Equilibrium if $\alpha \leq\frac{1}{n(n-1)-1}$ and $\alpha \geq 1 - \frac{1}{2n-1}$, respectively.
\end{theorem}
\begin{proof}
 We start with proving that the double clique network $DG_n$ where every agents owns an edge to all other agents is in Nash Equilibrium if $\alpha \leq \frac{1}{n(n-1)-1}$. 
 
 Clearly, in $DG_n$ no agent can improve its expected distance cost by buying one or more edges since each agent already has the minimum possible expected distance cost. The same holds true for performing edge-swaps. Swapping an edge yields a single towards some other node in the network which then yields an expected distance towards this node which is strictly larger than $1$. Since the edge-cost stays the same and the expected distance cost increases by swapping edges, no agent can improve by swapping one or more edges. 
 It remains to analyze edge-deletions. In $DG_n$ every agent has expected distance cost of $n-1$. Deleting $1\leq k \leq n-1$ edges yields expected distance cost of 
 $$\frac{(n(n-1)-2k)(n-1) + kn}{n(n-1)-k} = n-1 + \frac{k}{n(n-1)-k} \geq n-1 + \frac{1}{n(n-1)-1}.$$
 Thus, if $\alpha \leq \frac{1}{n(n-1)-1}$ then deleting one or more edges is not an improving move. 
 
 Next, we show that the double star $DS_n$ with arbitrary edge-ownership is in Nash Equilibrium if $\alpha \geq 1 - \frac{1}{2n-1}$. 
 
 Clearly, no agent can delete edges since this destroys the 2-edge-connectedness of the network which then induces infinite cost. Moreover, no agent can swap edges since this does not change the edge cost but increases the expected distance cost. Thus, we are left to analyze edge purchases. Clearly the center of the double star cannot buy edges to decrease its cost. Hence, we analyze edge purchases by non-center nodes of $DS_n$. 
 Every such agent has expected distance cost of $1+2(n-2) = 2n-3$.
 
 Let $u$ be a non-center agent and let $S_u$ be $u$'s current strategy in network $G = DS_n$. Assume that agent $u$ can change its strategy from $S_u$ to $S_u'$ and thereby strictly decrease its cost. Let $G'$ be the network $G$ after $u$'s strategy-change from $S_u$ to $S_u'$. We claim that if $\alpha > \frac{1}{n-1}$ and if $G'$ contains at least three edges between the center vertex and $u$ or if there are at least two edges between $u$ and some other non-center vertex $v$, then agent $u$ has a strategy $S_u''$, which strictly outperforms strategy $S_u'$ and where the corresponding network $G''$ has exactly two edges between the center vertex and $u$ and at most one edge between $u$ and any other non-center vertex. Thus, we can assume that if agent $u$ has an improving strategy-change, then there exists an improving strategy-change towards a strategy which buys only additional single edges towards other non-center vertices. After proving the above claim, we will prove that no such improving strategy-change exists if $\alpha \geq 1-\frac{1}{2n-1} > \frac{1}{n-1}$, which then implies that $DS_n$ is in Nash Equilibrium for all $\alpha \geq 1-\frac{1}{2n-1}$.  
 
 Now we prove the claim: We first show that strategy $S_u'$ can be improved if $G'$ contains at least three edges between the center vertex and $u$. In this case this implies that $u$ owns at least one edge towards the center vertex and that agent $u$ could remove one edge from $G'$ to ensure that at least two edges between its and the center vertex remain. Let $G''$ be the network $G'$ after the edge-removal and let $S_u''$ be the strategy $S_u'$ without the removed edge. This removal would save $\alpha$ in edge-cost. If $u$ has no single edges towards any non-center vertex, then its expected distance cost in $G''$ would not increase compared to its expected distance cost in $G'$ by the edge-removal since all edges on all its shortest paths are backed up by another parallel edge. Since $\alpha > 0$, $S_u''$ strictly outperforms strategy $S_u'$. If $u$ has $1\leq k \leq n-2$ single edges towards $k$ different non-center vertices in $G'$, then the edge-removal of one edge between $u$ and the center vertex increases the probability that one of the $k$ edges is destroyed by the adversary. The probability increases by $$\frac{k}{m(m-1)} \leq\frac{n-2}{2n(2n-1)} < \frac{1}{n-1},$$ where $m\geq 2n$ is the number of edges in $G'$. Thus, agent $u$'s expected distance cost in $G''$ increases by at most $\frac{1}{4n}$ compared to its expected distance cost in $G'$. Since $\alpha > \frac{1}{n-1}$, it follows that $S_u''$ strictly outperforms $S_u'$. If $G''$ contains more than three edges between $u$ and the center vertex, then we can apply the above argument iteratively to obtain a strategy $S_u''$ which strictly outperforms $S_u'$ and a corresponding network $G''$ which has exactly two edges between $u$ and the center vertex. 

 Now we show that strategy $S_u'$ can be improved if $G'$ contains at least two edges between $u$ and some other non-center vertex $v$. Note that in this case all edges between $u$ and $v$ are bought by agent $u$. 
 
 It is possible that in network $G'$ there is no edge or only one edge between $u$ and the center vertex. If there is no edge between $u$ and the center vertex, then agent $u$ could swap two edges from $v$ to the center vertex and thereby strictly decrease its cost. This is true since this swap would decrease $u$'s expected distance to every vertex $w\neq v$ by at least $1$ and it only increases its expected distance to $v$ by $1$. If there is exactly one edge between $u$ and the center vertex, then agent $u$ could swap one edge from $v$ to the center vertex and thereby decrease its cost. This swap may create a single edge towards $v$ but if this edge is attacked by the adversary then this only increases $u$'s distance to $v$ by $1$ whereas in $G'$ an attack on the single edge between $u$ and the center vertex increases $u$'s distances to $n-2$ vertices by at least $1$. Hence, if there is no edge or only one edge between $u$ and the center-vertex, then in both cases there is a strategy $S_u''$ which strictly outperforms strategy $S_u'$ and where the corresponding network $G''$ has exactly two edges between $u$ and the center vertex. 
 Thus, we will assume in the following that there are exactly two edges between $u$ and the center vertex and at least two edges between $u$ and some non-center vertex $v$.
 
 Let $G''$ be the network obtained from network $G'$ by removing one of the edges between $u$ and $v$ and let $S_u''$ be $u$'s strategy obtained by removing the mentioned edge from $S_u'$. If there are at least two edges between $u$ and $v$ in $G''$, then an analogous argument as above yields that $S_u''$ strictly outperforms $S_u'$ if $\alpha > \frac{1}{n-1} > \frac{n-3}{2n(2n+1)}$. Note that in this case $G'$ has at least $2n+1$ many edges. If there is a single edge between $u$ and $v$ in $G''$, then the number of non-center vertices to which $u$ has a single edge increases by $1$ from $k$ to $k+1$ for some $0\leq k\leq n-3$. Thus, its expected distance cost compared to network $G'$ increases by $$\frac{k+1}{2n-1}-\frac{k}{2n} = \frac{2n+k}{2n(2n-1)} \leq \frac{1}{2n-1} + \frac{n-3}{2n(2n-1)}  < \frac{1}{n-1},$$ which implies that $S_u''$ strictly outperforms strategy $S_u'$ if $\alpha > \frac{1}{n-1}$.  
  
 Having settled the claim, we now analyze the case where a non-center agent $u$ buys $1\leq k \leq n-2$ single edges to $k$ other non-center nodes. In this case $u$'s expected distance cost is
 \begin{align*} &\frac{2(n-1)(k+1+2(n-2-k)) + k(k+2(n-2-k+1))}{2(n-1)+k}\\ =& 2n-3 - k + \frac{k}{2(n-1)+k}.\end{align*}
 Since $-k+ \frac{k}{2(n-1)+k} \geq -1 + \frac{1}{2(n-1)+1} = -1 + \frac{1}{2n-1}$ for $1\leq k \leq n-2$ it follows that the expected distance cost after buying $1\leq k \leq n-2$ single edges is at least $2n-3 - 1 + \frac{1}{2(n-1)+1}$. Thus if $\alpha \geq 1 - \frac{1}{2n-1}$, then buying one or more single edges is not an improving move for any non-center agent.
 
 Since $\alpha \geq 1 - \frac{1}{2n-1} > \frac{1}{4n}$, it follows that no non-center vertex can buy one or more edges in network $DS_n$ to strictly decrease its cost. 
\end{proof}
\noindent Next, we show that NE in the Adv-NCG are not comparable with NE from the NCG or Kliemann's model.
\begin{theorem}\label{thm_compare}
There is a NE in the Adv-NCG which is not an NE in the NCG and vice versa. The analogous statement also holds for Kliemann's model. 
\end{theorem}
\begin{proof}
We compare the Nash equilibria of the Adv-NCG with the equilibria of the NCG~\cite{Fab03} and equilibria of Kliemann's model~\cite{Kli11}. Although the Adv-NCG can be understood as a mixture of both modes, we will show that the Nash equilibria of both models can be quite different from the Adv-NCG. Clearly, NE in the Adv-NCG are $2$-edge-connected, thus we should compare only $2$-edge-connected equilibria of all models. 

First we show that there is a $2$-edge-connected NE the NCG that is not a NE for the Adv-NCG, and vice versa.
\begin{figure}[h!]
\begin{minipage}[h]{0.5\linewidth}
	\center{\includegraphics[width=0.3\textwidth]{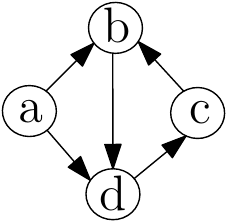}\\a) NE for Adv-NCG}
	\end{minipage}
	\hfill
	\begin{minipage}[h]{0.5\linewidth}
		\center{\includegraphics[width=0.3\linewidth]{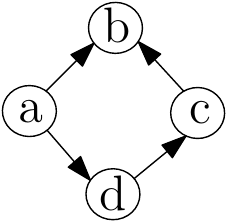} \\b) NE for NCG}
	\end{minipage}
 	\caption{Comparison of NE for NCG and Adv-NCG models}
 	\label{NE_for_M1}
\end{figure}  
Consider the network $G$ depicted in Fig~\ref{NE_for_M1}~a), 
is a NE for the Adv-NCG if $\frac{11}{15}\leq\alpha\leq\frac{7}{5}$. But this network is not a NE in the NCG, because agent $b$ could delete the edge $\{b,c\}$ and thereby decrease its cost from $\alpha+3$ to $4$ if $1 < \alpha \leq \frac{7}{5}$. 

For the reversed statement it is easy to see,that network $G'$, depicted in Fig.~\ref{NE_for_M1}~b) is a NE in the NCG, but is not a NE in the Adv-NCG if $1\leq\alpha\leq\frac{7}{5}$.

The main difference between Kliemann's model and the Adv-NCG is the 2-edge-connectedness. It means, that in Kliemann's model any agent has an individual cost $C_v(S)=|S_v|\cdot\alpha$, if the network is 2-edge-connected. It follows, that $DG_n$ is a NE in our model for very small $\alpha$, but $DG_n$ is not a NE in Kliemann's model.

Converse, for any small $\alpha$ a cycle where every agent owns exactly one edge is a NE in Kliemann's model, but it is not a NE for small $\alpha$ in the Adv-NCG.
\end{proof}

\subsection{Relation between the Diameter and the Social Cost}
We prove a property which relates the diameter of a network with its social cost. With this, we prove that one of the most useful tools for analyzing NE in the NCG~\cite{Fab03} can be carried over to the Adv-NCG.   

Before we start, we analyze the diameter increase induced by removing a single edge in a $2$-edge-connected network. 
\begin{lemma}\label{lem_diam_increase}
 Let $G  = (V,E)$ be any $2$-edge-connected network having diameter $D$ and let $G-e$ be the network $G$ where some edge $e\in E$ is removed. Then the diameter of $G-e$ is at most $2D$.
\end{lemma}
\begin{proof}
 Let $e = \{u,v\}$ be the edge which is removed from a diameter $D$ network $G$ to obtain the network $G-e$.  
 Consider any shortest path $P$ in $G$ which uses edge $e$ somewhere along the path. Let $x$ and $y$ be the endpoints of path $P$ and we assume that $u$ and $v$ are the endpoints of $e$ which are closer to $x$ and $y$, respectively.
 
 Since $G$ is $2$-edge-connected, we can find a smallest cycle $C$ in $G$ which includes edge~$e$. Let $z$ be a node in the cycle $C$ which has maximum distance to edge $e$, that is, maximum distance to both $u$ and $v$ simultaneously. There exists a shortest path $P_{xz}$ in $G$ which connects $x$ and $z$, and there is a shortest path $P_{zy}$ in $G$ which connects $z$ and $y$. Both paths have length at most $D$. Observe, that it is impossible that both paths $P_{xz}$ and $P_{zy}$ contain edge $e$, since otherwise there must be a shorter path between $x$ and $z$ or $z$ and $y$.  
 
 If both of the paths do not contain edge $e$, then $P_{xz}\cup P_{zy}$ is a path between nodes $x$ and $y$ in the network $G - e$ and it has a length at most $2D$.
 
 Finally, consider a situation where exactly one of the paths contains edge $e$ in $G$ and let this be path $P_{xz}$. Thus, we have that $P_{xz} = P_{xv}\cup \{v, u\} \cup P_{uz}$ or $P_{xz} = P_{xu}\cup\{u,v\}\cup P_{vz}$ in graph $G$ and both of the paths $P_{uz}$ and $P_{vz}$ are parts of the cycle $C$. By choice of $z$, it follows that $P_{xv}\cup P_{vz}$ or $P_{xu}\cup P_{uz}$ is a path between $x$ and $z$ in graph $G - e$ which has length at most $D$. Since the path $P_{zy}$ does not contain edge $e$, it follows that it can be used in graph $G - e$. Since $P_{zy}$ has length at most $D$, we have that the distance between nodes $x$ and $y$ in $G-e$ is at most $2D$. 
\end{proof}
\noindent Next, we will focus on edges which are part of cuts of the network of size two. Remember that a bridge is an edge whose removal from a network increases the number of connected components of that network. Let $G = (V,E)$ be any $2$-edge-connected network. We say that an edge $e \in E$ is a \emph{$2$-cut-edge} if there exists a cut of $G$ of size $2$ which contains edge $e$. Equivalently, $e$ is a $2$-cut-edge of $G$ if its removal from $G$ creates at least one bridge in $G-e$. 
We now bound the number of $2$-cut-edges in any $2$-edge-connected network $G$. This is an important structural result, since this proves that the amount of edge-overbuilding due to the adversary is sharply limited. 
\begin{lemma}\label{lem_number_of_2_cut_edges}
 Any $2$-edge-connected network $G$ with $n$ nodes can have at most $2(n-1)$ edges which are $2$-cut-edges.
\end{lemma}
\begin{proof}
Let $e$ be any $2$-cut-edge in network $G$. By definition, the removal of $e$ creates one or more bridges in $G-e$. Let $b_1,\dots,b_l$ denote those bridges. Note, that $b_1,\dots,b_l$ also must be $2$-cut-edges in $G$. Moreover, it follows that there must be a shortest cycle $C$ in $G$ which contains all the edges $e,b_1,\dots,b_l$. If there are more than one such cycles, then fix one of them. We call the fixed cycle $C$ a cut-cycle. 

Notice that any 2-cut-edge corresponds to exactly one cut-cycle in the network and that every cut-cycle contains at least two 2-cut-edges. We show in the following that if any cut-cycle in the network contains at least three 2-cut-edges, then we can modify the network to obtain strictly more 2-cut-edges and strictly more cut-cycles. This implies that the number of 2-cut-edges is maximized if the number of cut-cycles is maximized and every cut-cycle contains exactly two 2-cut-edges.  

Now we describe the procedure which converts any network with at least one cut-cycle containing at least three 2-cut-edges into a modified network with a strictly increased number of cut-cycles and 2-cut-edges (see Fig. \ref{Transf_to_max_2ce}). 
\begin{figure}[h!]
	\begin{minipage}[h]{0.5\linewidth}
		\center{\includegraphics[width=0.9\linewidth]{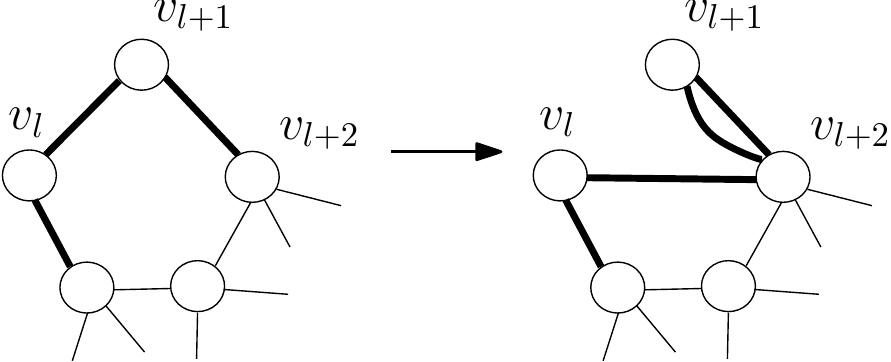}\\a) Conversion of adjacent 2-cut-edges}
	\end{minipage}
	\hfill
	\begin{minipage}[h]{0.5\linewidth}
		\center{\includegraphics[width=1\linewidth]{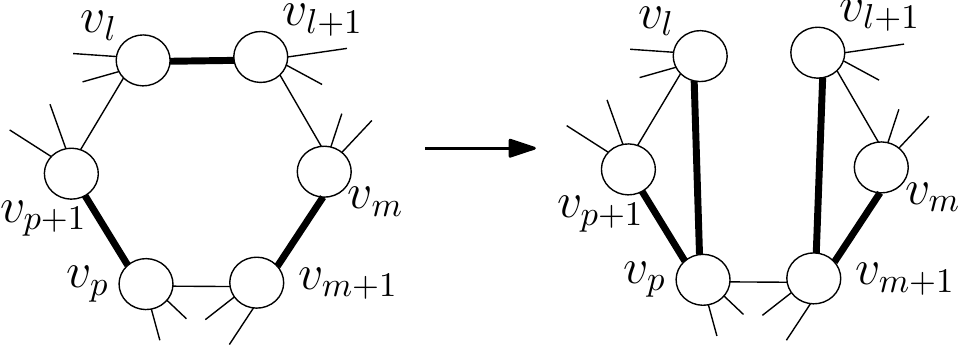} \\b) Conversion of non-adjacent 2-cut-edges}
	\end{minipage}
	
 	\caption{Increasing the number of 2-cut-edges by splitting up a cut-cycle.}
 	\label{Transf_to_max_2ce}
\end{figure}

Let $G$ be any network with at least one cut-cycle $C = v_1, \ldots, v_k, v_1$ containing at least three 2-cut-edges. 
If there are two adjacent 2-cut-edges $\{v_l, v_{l+1}\}, \{v_{l+1}, v_{l+2}\}$ in cycle $C$, then delete the 2-cut-edge $\{v_l, v_{l+1}\}$ and insert two new edges $\{v_l, v_{l+2}\}$ and $\{v_{l+1}, v_{l+2}\}$. First of all, note that these new edges have not been present in network $G$ before the insertion since otherwise $\{\{v_l,v_{l+1}\},\{v_{l+1},v_{l+2}\}\}$ cannot be a cut of $G$. We claim that both new edges are 2-cut-edges and that the cycle $C$ is divided into two new cut-cycles $v_1, \ldots, v_l, v_{l+2}, \ldots, v_k, v_1$ and $v_{l+1}, v_{l+2}, v_{l+1}$. Indeed, there are at least two bridges $\{v_{l+1}, v_{l+2}\}$ and $\{v_{k}, v_{k+1}\}$ in the cut-cycle $C$ after deleting $\{v_l, v_{l+1}\}$, and both of them end up in different new cut-cycles. Hence, deleting any of the newly inserted edges $\{v_l, v_{l+2}\}$ or $\{v_{l+1}, v_{l+2}\}$ implies that $\{v_{k}, v_{k+1}\}$ or $\{v_{l+1}, v_{l+2}\}$ becomes a bridge. Thus, both new edges are 2-cut-edges and both of new cycles are cut-cycles.

If there are three pairwise non-adjacent 2-cut-edges $\{v_l, v_{l+1}\}$, $\{v_m, v_{m+1}\}$, $\{v_p, v_{p+1}\}$ in cycle $C$, then delete one 2-cut-edge $\{v_l, v_{l+1}\}$ and insert two new edges $\{v_l, v_{p+2}\}$ and $\{v_{l+1}, v_{m+1}\}$. Analogous to above, both new edges cannot be already present in $G$ and both are 2-cut-edges because deleting any of them renders edge $\{v_m, v_{m+1}\}$ or $\{v_p, v_{p+1}\}$ a bridge. Moreover, cut-cycle $C$ is divided into two new cut-cycles.

Finally, we claim that the maximum number of cut-cycles in any $n$-vertex network $G$ is at most $n-1$. Since we know that every such cut-cycle contains exactly two 2-cut-edges this then implies that there can be at most $2(n-1)$ 2-cut-edges in any network $G$.

Now we prove the above claim. Note that applying our transformation does not disconnect the network. Thus, we know that network $G$ after all transformations is connected. Now we iteratively choose any cut-cycle $C$ in $G$ and we delete the two 2-cut-edges contained in $C$. This deletion increases the number of connected components of the current network by exactly $1$. We repeat this process until we have destroyed all cut-cycles in $G$. Note that deleting edges from $G$ may create new cut-cycles, but we never destroy more than one of them at a time. Thus, since each iteration increases the number of connected components of the network by $1$, it follows that there can be at most $n-1$ iterations since network $G$ with $n$ vertices cannot have more than $n$ connected components.
\end{proof}
\begin{remark}
Lemma~\ref{lem_number_of_2_cut_edges} is tight, since a path of length $n-1$, where all neighboring nodes are connected via double edges, has exactly $2(n-1)$ $2$-cut-edges. 
\end{remark}
\noindent Now we relate the diameter with the social cost.
\begin{theorem}\label{thm_diam_social_cost}
 Let $(G,\alpha)$ be any NE network on $n$ nodes having diameter $D$ and let $OPT(n,\alpha)$ be the optimum network on $n$ nodes for the same edge-cost $\alpha$. Then we have that $$\frac{cost(G,\alpha)}{cost(OPT(n,\alpha))} \in \mathcal{O(D)}.$$
\end{theorem}
\begin{proof}
 Since $OPT(n,\alpha)$ must be $2$-edge-connected, it must have at least $n$ edges. Moreover, the minimum expected distance between each pair of vertices in $OPT(n,\alpha)$ is at least $1$. Thus, we have that $cost(OPT(n,\alpha)) \in \Omega(\alpha \cdot n + n^2)$.
 
 Now we analyze the social cost of the NE network $(G,\alpha)$, where $G = (V,E)$. We have $cost(G,\alpha) = edge(G,\alpha) + dist(G)$ and we will analyze both terms separately. We start with an upper bound on $dist(G)$.
 
 Since $(G,\alpha)$ has diameter $D$ and since $(G,\alpha)$ is $2$-edge-connected, Lemma~\ref{lem_diam_increase} implies that the expected distance between each pair of vertices in $(G,\alpha)$ is at most $2D$. Thus, we have that $dist(G) \in \mathcal{O}(n^2\cdot D).$

 Now we analyze $edge(G,\alpha)$. By Lemma~\ref{lem_number_of_2_cut_edges} we have at most $2n$ many $2$-cut-edges in $G$. Buying all those edges yields cost of at most $2n\cdot \alpha$. 
 
 We proceed with bounding the number of non-$2$-cut-edges in $G$. We consider any agent $v$ and analyze how many non-$2$-cut-edges agent $v$ can have bought. We claim that this number is in $\mathcal{O}\left(\frac{nD}{\alpha}\right)$, which yields total edge-cost of $\mathcal{O}(nD)$ for agent $v$. Summing up over all $n$ agents, this yields total edge-cost of $\mathcal{O}(n^2D)$ for all non-$2$-cut-edges of $G$. This then implies an upper bound of $\mathcal{O}(\alpha\cdot n + n^2D)$ on the social cost of $G$ which finishes the proof.
 
 Now we prove our claim. Fix any non-$2$-cut-edge $e = \{v,w\}$ of $G$ which is owned by agent $v$. Let $V_e \subset V$ be the set of nodes of $G$ to which all shortest paths from $v$ traverse the edge $e$.  
 
 We first show that removing the edge $e$ increases agent $v$'s expected distance to any node in $V_e$ to at most $4D$. By Lemma~\ref{lem_diam_increase}, removing edge $e$ increases the diameter of $G$ from $D$ to at most $2D$. Since $e$ is a non-$2$-cut-edge, we have that $G-e$ is still $2$-edge-connected. Thus, again by Lemma~\ref{lem_diam_increase}, it follows that agent $v$'s expected distance to any other node in $G-e$ is at most $4D$.
 
 However, removing edge $e$ not only increases $v$'s expected distance towards all nodes in $V_e$, instead, since $G-e$ has a less many edges than $G$, agent $v$'s expected distance to \emph{all} other nodes in $V\setminus(V_e\cup \{v\})$ increases as well. We now proceed to bound this increase in expected distance cost.   
 
 We compare agent $v$'s expected distance cost in network $G$ and in network $G-e$. Let $m$ denote the number of edges in $G$. Thus, $G-e$ has $m-1$ many edges. For network $G$ agent $v$'s expected distance cost is $$dist_G(v) = \frac{1}{m}\sum_{f\in E}\delta_{G-f}(v) = \frac{1}{m} \sum_{f\in E\setminus\{e\}}\delta_{G-f}(v) + \frac{\delta_{G-e}(v)}{m}.$$ In network $G-e$, we have 
 $dist_{G-e}(v) = \frac{1}{m-1}\sum_{f\in E\setminus\{e\}}\delta_{G-e-f}(v)$.
 Now we upper bound the increase in expected distance cost for agent $v$ due to removal of edge $e$ from $G$. $dist_{G-e}(v) - dist_G(v)$ is
 \begin{align*}
  &  \frac{1}{m-1}\sum_{f\in E\setminus\{e\}}\delta_{G-e-f}(v) - \left(\frac{1}{m} \sum_{f\in E\setminus\{e\}}\delta_{G-f}(v) + \frac{\delta_{G-e}(v)}{m}\right) \\
  =& \sum_{f\in E\setminus\{e\}}\left(\frac{\delta_{G-e-f}(v)}{m-1} - \frac{\delta_{G-f}(v)}{m}\right) -\frac{\delta_{G-e}(v)}{m}.
 \end{align*}
 We have that $\delta_{G-e-f}(v) \leq \delta_{G-f}(v) + |V_e|\cdot 4D$, since in $G-e-f$ only the distances to nodes in $V_e$ increase, compared to the network $G-f$ and since $e$ is a non-$2$-cut-edge in $G$. Moreover, by Lemma~\ref{lem_diam_increase}, the distances to nodes in $V_e$ in $G-e-f$ increase to at most $4D$ for each node in $V_e$.
 Thus, we have that $dist_{G-e}(v) - dist_G(v)$ is 
 \begin{align*}
  &\sum_{f\in E\setminus\{e\}}\left(\frac{\delta_{G-e-f}(v)}{m-1} - \frac{\delta_{G-f}(v)}{m}\right) -\frac{\delta_{G-e}(v)}{m}\\
  &\leq  \sum_{f\in E\setminus\{e\}}\left(\frac{\delta_{G-f}(v) + |V_e|4D}{m-1} - \frac{\delta_{G-f}(v)}{m}\right)\\
  &= |V_e|4D + \sum_{f\in E\setminus\{e\}}\left(\frac{\delta_{G-f}(v)}{m(m-1)}\right) \leq |V_e|4D + \sum_{f\in E\setminus\{e\}}\left(\frac{2D\cdot n}{n(m-1)}\right)\\
  &\leq |V_e|4D + 4D  = (|V_e|+1)4D.
 \end{align*}
 Since $G$ is in Nash Equilibrium, we know that removing edge $e$ is not an improving move for agent $v$. Thus, we have that $$\alpha \leq  (|V_e|+1)4D \iff  |V_e| \geq \frac{\alpha}{4D}-1.$$
 Thus, for all non-$2$-cut-edges $e$ which are bought by agent $v$, we have that $|V_e|\in \Omega(\frac{\alpha}{D})$. Since all these sets $V_e$ are disjoint, it follows that $v$ can have bought at most $\frac{n}{\Omega(\frac{\alpha}{D})} \in \mathcal{O}(\frac{nD}{\alpha})$ many non-$2$-cut-edges. 
\end{proof}

\subsection{Price of Stability and Price of Anarchy}
\begin{theorem}\label{thm_PoS}
If $\alpha\leq\frac{1}{n(n-1)-1}$, then the PoS is 1. If $\frac{1}{n(n-1)-1} < \alpha < \frac{2}{n(n-1)-1}$, then PoS is strictly larger than $1$, if  $\alpha>1-\frac{1}{2n-1}$, then the PoS is at most 2. 
\end{theorem}
\begin{proof}
By Theorem~\ref{thm_existence} and Observation~\ref{obs_opt} $DG_n$ network is optimal and is a Nash equilibrium when $\alpha\leq\frac{1}{n(n-1)-1}$. Thus, price of stability is 1 for this value of $\alpha$.

The second statement holds, since $DG_n$ is the unique optimum for $\alpha \leq \frac{2}{n(n-1)-1}$ but any agent could delete an edge and thereby increase its expected distance cost by $\frac{1}{n(n-1)-1}$. Thus, if $\frac{1}{n(n-1)-1} < \alpha \leq \frac{2}{n(n-1)-1}$, then this edge-deletion is an improving move which shows that $DG_n$ is not a NE. 

The third statement follows from Theorem~\ref{thm_existence} and the simple lower bound on the expected social cost of the optimum from the proof of Theorem~\ref{thm_diam_social_cost}. 
Thus, for $\alpha>1-\frac{1}{2n-1}$ the PoS is at most $\frac{cost(DS_n)}{n\alpha+n^2}=\frac{2(n-1)\alpha+2(n-1)^2}{n\alpha+n^2}\leq 2$.
\end{proof}
\noindent We now show how to adapt two techniques from the NCG for bounding the diameter of equilibrium networks to our adversarial version. This can be understood as a proof of concept showing that the Adv-NCG can be analyzed as rigorously as the NCG. However, carrying over the currently strongest general diameter bound of $2^{\mathcal{O}(\sqrt{\log n})}$ due to Demaine et al.~\cite{De07}, which is based on interleaved region-growing arguments seems challenging due to the fact that we can only work with expected distances.   

We start with a simple diameter upper bound based on \cite{Fab03}. 
\begin{theorem}
 The diameter of any NE network $(G,\alpha)$ is in $\mathcal{O}(\sqrt{\alpha})$.
\end{theorem}
\begin{proof}
 We prove the statement by contradiction. Assume that there are agents $u$ and $v$ in network $G$ with $d_G(u,v)\geq 4\ell$, for some $\ell$. Since expected distances cannot be shorter than distances in $G$, it follows that $u$'s expected distance to $v$ is at least $4\ell$. If $u$ buys an edge to $v$ for the price of $\alpha$ then $u$'s decrease in expected distance cost is at least $\frac{|E|}{|E|+1}(4\ell-1+4\ell-3+\cdots+1)= \frac{|E|}{|E|+1}2\ell^2$.  
 
 Thus, if $d_G(u,v)>4\sqrt{\alpha}$, then $u's$ decrease in expected distance cost by buying the edge $uv$ is at least $\frac{|E|}{|E|+1}2\alpha > \alpha$. Thus, if the diameter of $G$ is at least $4\sqrt{\alpha}$, then there is some agent who has an improving move. 
\end{proof}
\noindent Together with Theorem~\ref{thm_diam_social_cost} this yields the following statement:
\begin{corollary}
 The Price of Anarchy of the Adv-NCG is in $\mathcal{O}(\sqrt{\alpha})$.
\end{corollary}
\noindent Next, we show how to adapt a technique by Albers et al.~\cite{Al14} to get a stronger statement, which implies constant PoA for $\alpha \in \mathcal{O}(\sqrt{n})$.
\begin{theorem}\label{thm_PoA}
The Price of Anarchy of the Adv-NCG is in $\mathcal{O}\left(1 + \frac{\alpha}{\sqrt{n}}\right)$.
\end{theorem}
\begin{proof}
We use Theorem~\ref{thm_diam_social_cost} and give an improved bound on the expected diameter of any NE network. Let the expected diameter of the network be $d$ and consider nodes $u$ and $v$ which have expected distance $d$.

Let $B$ be the set of nodes in the network which are at expected distance of $d'= \left\lfloor\frac{d-1}{8}\right\rfloor$ from node $u$. First, we analyze the change in expected distance cost of agent $v$ if it buys an edge towards $u$. Consider any node $w\in B$. By Lemma~\ref{lem_diam_increase} we have that without edge $\{v,u\}$ agent $v$ has expected distance of at least $\frac{d}{2}-d'$ towards $w$. After buying the edge $\{v,u\}$, $v$'s expected distance to $w$ is at most $\frac{(1+d')|E|+d}{|E|+1}$. Thus, $v$'s expected distance to $w$ decreases by at least $$\frac{d}{2}- d'- \left(\frac{|E| d'+|E|+d}{|E|+1}\right) \geq \frac{d}{2}- 2d'-2 > \frac{d-8}{4}.$$
It follows that by buying the edge $vu$ agent $v$'s expected distance cost decreases by at least $\frac{d-8}{4}|B|$.
Since $G$ is in NE, it follows that $\alpha \geq (\frac{d-8}{4})|B|$.

Now consider node $u$ which has expected distance of at most $d'$ to any node $B$. Thus, by Lemma~\ref{lem_diam_increase}, and since $d_G(u,v) \geq \frac{d}{2}$, we know that there must be nodes $w\in B$ with $d_G(u,w) = \frac{d'}{2}$. Let $D$ denote set of all nodes $w$ in $B$ with $d_G(u,w) \leq \frac{d'}{2}$. For any node $w\in D$ let $$S_w=\{x \mid w \text{ is the last node in $D$ on a shortest path from $u$ to $x$}\}.$$ 
If $S_w$ is non-empty, then $d_G(u,w)$ is  $\frac{d'}{2}$. Since there are $n-|D|$ nodes outside of $D$ it follows, that there must be some node $w$ with $|S_w| \geq \frac{n-|D|}{|D|}$. If $u$ buys the edge $\{u,w\}$, then $u$'s expected distance cost decreases by at least $$\left(\frac{d'}{2}-\frac{|E| + \frac{d'}{2}}{|E|+1}\right)|S_w| \geq \left(\frac{d'}{4}-\frac{1}{2}\right)|S_w|.$$ Since $G$ is an NE, it follows that $\alpha \geq \left(\frac{d'}{4}-\frac{1}{2}\right)|S_w| \geq \left(\frac{d'}{4}-\frac{1}{2}\right)\frac{n-|D|}{|D|}$. By rearranging we get $$
 |D|2\alpha \geq |D|\left(\alpha + \left(\frac{d'}{4}-\frac{1}{2}\right)\right) \geq \left(\frac{d'}{4}-\frac{1}{2}\right)n,$$ where the first inequality holds since $\alpha \geq d > \frac{d'-2}{4}$ because $G$ is in NE. Thus, we have $|B|\geq |D| \geq (d'-2)\frac{n}{8\alpha}$. 
        
From $\alpha \geq (\frac{d-8}{4})|B|$, we get $\alpha \geq \left(\frac{d-8}{4}\right)\left(\frac{d'-2}{8\alpha}\right)n \iff 8\alpha^2 > (\frac{d}{2}-2)(d'-2)n$. Since $\frac{d}{2} > d'$ we have $$8\alpha^2 \geq (d'-2)^2 n \iff \sqrt{\frac{8}{n}}\alpha \geq (d'-2) \geq (\frac{d-1}{8}-3).$$ Hence, we have $25 + \frac{8\sqrt{8}\alpha}{\sqrt{n}} \geq d.$         
\end{proof}

\begin{theorem}\label{thm_PoAlower}
 The Price of Anarchy of the Adv-NCG is at least $2$ and for very large $\alpha$ this bound is tight.
\end{theorem}
\begin{proof}
 Consider an arbitrary large $\alpha$, e.g. $\alpha  = 2^n$. In that case the optimum network must be a cycle whereas, by Theorem~\ref{thm_existence}, the network $DS_n$ is in Nash Equilibrium for this $\alpha$. Since $DS_n$ has $2(n-1)$ edges and since in this range of alpha the edge-cost term dominates the social cost, the lower bound follows. The tight upper bound for large $\alpha$ follows from Lemma~\ref{lem_number_of_2_cut_edges}, since $2$-cut-edges cannot be deleted without creating a bridge. 
\end{proof}
\section{Conclusion}
Our work is the first step towards incorporating both centrality and robustness aspects in a simple and accessible model for selfish network creation. In essence we proved that many properties and techniques can be carried over from the non-adversarial NCG and we indicated that the landscape of optimum and equilibrium networks in the Adv-NCG is much more diverse than without adversary. As for the NCG, proving strong upper or lower bounds on the PoA is very challenging. Especially surprising is the hardness of constructing higher lower bounds than in the NCG since by introducing suitable gadgets it is always possible to enforce that no agent wants to swap or delete edges. A non-constant lower bound on the PoA seems possible if $\alpha$ is linear in $n$. 

It would also be interesting to consider different adversaries. An obvious candidate for this is node-removal at random. Another promising choice is a local adversary, where every agent considers that some of its incident edges may fail. This local perspective combined with a centrality aspect could explain why many selfishly built networks have a high  clustering coefficient. 

Another direction is to consider the swap version~\cite{MS10,Ehs15} of the Adv-NCG, especially in the case where all agents own at least $2$ edges. We note in passing, that the swap-version of the Adv-NCG is not a potential game. The following improving move cycle shows that even if agents are only allowed to perform multi-swaps, then infinite sequences of improving moves are possible.
\begin{figure}[h!]
 \centering
 \includegraphics[width=\textwidth]{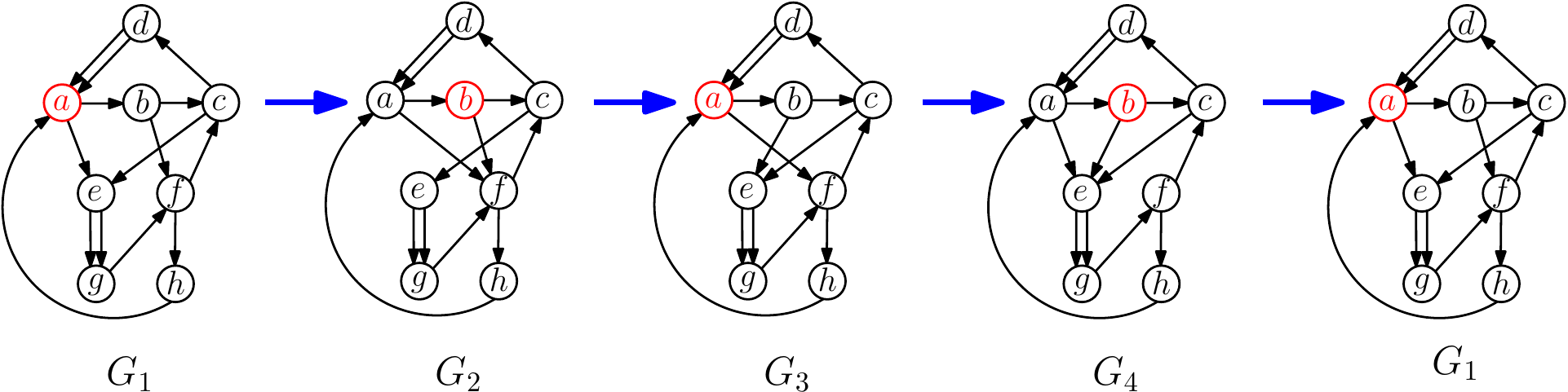}
 \caption{Improving Move Cycle for Swap Version}
 \label{fig:swap_IRC}
\end{figure}

\noindent Moreover creating equilibrium networks having diameter $4$ is already very challenging.

\bibliographystyle{abbrv}
\bibliography{robust_NCG}

\newpage
\appendix
\section{Appendix}
\begin{theorem}Given graph $G(V,E)$ \label{Thm_Hardness}
\begin{itemize}
\item[1.]  It is $NP$-hard to compute $\min12cds$.
\item[2.] It is $W[2]$-hard to compute $\min12cds$ parameterized by solution set $S$.
\end{itemize}
\end{theorem}
\begin{proof}

Let  $V=\{v_1,v_2,...,v_n\}$ and $E$  vertex set of a graph $G(V,E)$,  produce an new instance $G'(V',E')$ where, $V'=V\cup u$ and $E'=E\bigcup_{i\in [n]}\{v_i,u\}$.
\begin{itemize}
\item[(1)] We prove hardness by giving polynomial time reduction to \textsc{Dominating-Set} problem which is well known NP-complete problem. 

Now consider the $\min12cds$ $S$ of $G'$. We have $|S| = |D|+1$, where $D$ is a minimum dominating set of $G$. This holds, since if $|S|> |D|+1$, then we could simply replace $S$ by $D\cup \{u\}$ and obtain a smaller $\min12cds$. If $|S|<|D|+1$, then, since in $S$ every node is connected to two dominating nodes and since $S$ is connected, we can delete an element from $S$ to obtain a dominating set $D'$ with $|D'|<|D|$.  

Now consider $G'$ there will can be only two possibilities for the solution set $S$ of $\min12cds$ of $G'$:

\begin{itemize}
\item[(a)] If $S$ contains $u$, then $S \setminus\{u\}$ is a minimum dominating set of $G$.
\item[(b)] If $S$ does not contain $u$, then we can delete any element from $S$ to obtain a minimum dominating set of $G$.
\end{itemize} 
\item[(2)] We prove $W[2]$ hardness of the $\min12cds$ parameterized by solution set $S$ by giving a $fpt$-reduction to parameterized \textsc{Dominating-Set} which is known to be $W[2]$-complete. 

For the reduction, create input instance $(G',k+1)$ for $\min12cds$. From the proof of (1) we know that $S$ always contains a dominating set of $G$. Thus if there exist a $\min12cds$ of size at most $k+1$ then there exist a dominating set in $G$ of size at most $k$. 
\end{itemize}
\end{proof}

\end{document}